\documentclass{llncs}

\usepackage{latexsym}

\usepackage{amsmath, amsthm}
\usepackage{upgreek}
\usepackage{tikz}
\usepackage{times}
\usepackage{amssymb}
\usepackage[font={small}]{caption}
\usepackage{subcaption}
\usepackage{textcomp} 
\usepackage[font={small},caption=false]{subfig}
\usepackage{mathtools}
\usepackage{graphicx}
\usetikzlibrary{automata,positioning}

\makeatletter
\newtheorem*{rep@theorem}{\rep@title}
\newcommand{\newreptheorem}[2]{%
\newenvironment{rep#1}[1]{%
 \def\rep@title{#2 \ref{##1}}%
 \begin{rep@theorem}}%
 {\end{rep@theorem}}}
 
\def\@seccntformat#1{\@ifundefined{#1@cntformat}%
   {\csname the#1\endcsname\quad}  % default
   {\csname #1@cntformat\endcsname}% enable individual control
}
 
\let\oldappendix\appendix %% save current definition of \appendix
\renewcommand\appendix{%
    \oldappendix
    \newcommand{\section@cntformat}{\appendixname~\thesection.\quad}}
\makeatother

\newtheorem{transformation}{Transformation}
\newreptheorem{theorem}{Theorem}
\newreptheorem{lemma}{Lemma}

\newcommand*\Heq{\ensuremath{\overset{\kern2pt H}{=}}}

\newcommand\pepaeq{\mathrel{\stackrel{\makebox[0pt]{\mbox{\normalfont\tiny def}}}{=}}}

\usepackage{url}
\urldef{\mails}\path|{m.talebi, j.f.groote, j.p.linnartz}@tue.nl|    
\newcommand{\keywords}[1]{\par\addvspace\baselineskip
\noindent\keywordname\enspace\ignorespaces#1}

\begin{document}

\title{Communication Patterns in Mean Field Models for Wireless Sensor Networks}
\author{Mahmoud Talebi, Jan Friso Groote, Jean-Paul Linnartz}

\institute{Eindhoven University of Technology, Eindhoven, Netherlands\\
\mails\\}

\maketitle

\begin{abstract}
Wireless sensor networks are usually composed of a large number of nodes, and with the increasing processing power and power consumption efficiency they are expected to run more complex protocols in the future. These pose problems in the field of verification and performance evaluation of wireless networks. In this paper, we tailor the mean-field theory as a modeling technique to analyze their behavior. We apply this method to the slotted ALOHA protocol, and establish results on the long term trends of the protocol within a very large network, specially regarding the stability of ALOHA-type protocols.

\end{abstract}

\section{Introduction}
The Internet of Things requires the connection of many nodes, therefore we expect that in the future we will build Wireless Sensor Networks (WSNs) with a very large number of nodes. The analysis of such systems dates from 1970's but still many problems have not been solved. In particular, the behavior of such networks under heavy traffic loads requires improvements. Typical theoretical studies can only cover a limited number of aspects, while conclusions from simulations of large networks are often hard to generalize. 

In this paper we aim at developing a better understanding of the overall health of the network, for instance in terms of the number of nodes that remain in backlog; that is, those nodes which have been unsuccessful in delivering their messages and keep trying to retransmit the same message. In the past, such analyses have been presented, e.g. \cite{namislo1984analysis,van1990stability}, but these studies modeled ALOHA, which is a specific radio protocol that supports communication from many nodes to a single central base station. We witness a need to model the network protocols in more detail \cite{dandelski2015scalability}. 

In fact we observed that mean-field theory can be a powerful method to extend the approach in \cite{namislo1984analysis} and \cite{van1990stability}. Two of the most important ways in which mean-field theory outperforms the explicit Markov chain analysis is that it does not depend on the size of the network, even so, it gives better approximation when the numbers are very large; and at the same time with more complex specifications in protocols, the size of the semantic model only grows linearly in terms of the number of equations.

In this paper we show that mean-field theory can reconfirm the results in \cite{carleial1975bistable} and \cite{van1990stability}, and aim to show that the mean field approximation method can be suitable to model networks. In particular, we focus on an important behavior which occurs in networks, called \textit{bistability}, i.e., the situation where the network may stay for prolonged periods of time either in a good state with favorable performance or in an undesirable state in which many nodes keep repeating messages, but their transmission are lost because of mutual interference. That is, the system can converge to two or more different steady states, each of which would be a valid solution for the system. To the best of our knowledge such behavior has not yet been studied using the mean field approximations up until now.

The mean-field theory originally took shape as a method of approximating complex stochastic processes in physics, and was first formally described for communication systems in \cite{benaim2008class}. However, similar heuristics have been proposed before in \cite{hillston2005fluid} which conformed with steps taken while doing the analysis by means of a continuous time Markov chain, in the context of stochastic process algebras. Steps to apply these approaches to protocols in communication networks have also been taken e.g., in \cite{bakhshi2010automating}. Apart from these more formal and more structured approaches, mean-field analysis has also been applied to the analysis of exponential back-off algorithms in communication protocols \cite{bianchi2000performance,bordenave2005random}, and to study end-to-end delay and throughput in ad hoc networks together with non-equilibrium methods which focus on modeling specific topologies \cite{srinivasa2012statistical}.

For this text, we are going to apply the methods as discussed in \cite{bortolussi2013continuous} to modeling wireless sensor networks and do a detailed modeling of phenomena that are observable in such networks. The structure of the rest of this paper is as follows: in \textit{section \ref{sec:method}}, we give a concise explanation of the essential theory of the mean field approximation method, and then in \textit{section \ref{sec:interference}} we first explain interference in networks and then apply it to a network running the slotted ALOHA protocol in a simple setting, and consequently show its bistability. In \textit{section \ref{sec:broadcast}}, we first describe a way to include broadcasting and then consider a slotted ALOHA network with a more complex node behavior. Finally, in \textit{section \ref{sec:vector}} we briefly point to the use of vector fields in identifying the conditions in which a network shows stability in its behavior.

\section{Mean-field approximation}\label{sec:method}
In this section, we describe the necessary concepts of the mean field approximation method. We consider that every system consists of a number of components. Components are formally defined as follows:

\iffalse
Each component is capable of doing three types of actions: \textit{send actions}, \textit{receive actions}, and \textit{internal actions}. The internal actions can be carried out independently in each node. The send and receive actions however, are linked to each other; one or more nodes receives a message only when one or more nodes send a message. Each action also has an associated \textit{rate} which is generally a real number. However, receive actions have no number associated to them, and their rate is written as $\bot$; meaning that the rate of them happening depends on their communication with other parts of the system. The view on action types in this work derives heavily from \cite{hillston2005compositional}.
\fi

\begin{definition}
(Component). A component is a quadruple $\mathcal{C}=(S, Act, \rightarrow, s_0)$, where:
\begin{itemize}
\item $S$ is a nonempty set of states.
\item $Act$ is a set of action names.
\item $\rightarrow\subseteq S \times Act \times R \times S$ is a transition relation with $R~=~\mathbb{R}^{>0}\cup \{\bot\}$ called the set of transitions rates.
\item $s_0 \in S$ is the initial state.
\end{itemize}
\end{definition}

We also write a transition $(s, \alpha, r, s') \in \rightarrow$ where $\alpha \in Act$, $r \in R$ and $s, s' \in S$ as $s \xrightarrow{(\alpha, r)} s'$. A more descriptive name for the concept of component would be component transition system, since it closely resembles transition systems and other familiar automata.

For a component $\mathcal{C}$, at any point in its computation it can be in a state $Y\in S$ which is called the state of the component. Now let us consider a system composed of $N$ identical components.

\begin{definition}
(System). Let $\mathcal{C}_i=(S, Act, \rightarrow, s_{0_{i}})$ where $1\leq i\leq N$ be identical components, except for their initial conditions $s_{0_i}$. A system $\mathcal{M}$ consisting of these components is defined as:
\[\mathcal{M}\pepaeq(\mathcal{C}_{1}~||\ldots||~\mathcal{C}_{N}).\]

Where $||$ is a parallel operator. Taking the state of each component $\mathcal{C}_i$ to be described by a corresponding $Y_i\in S_i$, we define the state of the system as:
\begin{equation}\label{eq:systemstate}
\langle Y_1,\ldots,Y_N\rangle.
\end{equation}
\end{definition}

\iffalse
The main focus of this paper is the definition of the parallel operator ($||$).  When a component is only doing its internal actions, this operator shows the parallel execution of the components within the system. However, when a send action or a receive action are being executed, there can be communication between different components of the system. This operator is defined more in detail in the following sections.
\fi

Since all components are identical, the mean-field theory allows the modeling of system behavior on an abstract level. The overall state of the system is equivalent for different permutations of values in the vector (\ref{eq:systemstate}), and the state of the system can be represented in a \textit{numerical vector form}.

We take a system model $\mathcal{M}$ with $N$ components $\mathcal{C}_{i}, 1\leq i\leq N$, each with $n$ states. The \textit{numerical vector form} of $\mathcal{M}$, $X^{(N)}_\mathcal{M}$ is a vector with $n$ entries which is an alternative representation of the state of the system
\[X_\mathcal{M}^{(N)}=\langle X_{1}^{(N)},\ldots,X_{n}^{(N)}\rangle.\]

Here we explicitly mention $N$ in the notation, since in general we take the size of a system as a variable. The entry $X_{j}^{(N)}$ in vector $X_\mathcal{M}^{(N)}$ records the number of components which are currently in state $j$. Quick observations are that $0\leq X_{j}^{(N)}\leq N$ and $\sum_{j\in \{1,\ldots,n\}}X_{j}^{(N)}=N$.

Next we describe a normalized vector $\hat{X}_\mathcal{M}^{(N)}$ derived from $X_\mathcal{M}^{(N)}$ to be the following:
\[\hat{X}_\mathcal{M}^{(N)}=\langle \hat{X}_{1}^{(N)},\ldots,\hat{X}_{n}^{(N)}\rangle\]
where for each entry $\hat{X}_{j}^{(N)}=\frac{X_{j}^{(N)}}{N}$. The entries of this vector are also called \textit{occupancy measures}.

Based on this latter representation of the state of a system we define a class of automata called Population Continuous Time Markov chains, which are based on the normalized Population Continuous Time Markov Chain models (PCTMCs) with system size $N$, described in \cite{bortolussi2013continuous}.

\begin{definition}
(Normalized PCTMC model). Let $\mathcal{M}$ be a system of $N$ identical components $\mathcal{C}$, where the number of states of $\mathcal{C}$ is $n$. A normalized Population Continuous Time Markov Chain model for $\mathcal{M}$ is a triple $\hat{\chi}_{\mathcal{M}}^{(N)}=(\hat{X}^{(N)}, \mathcal{F}^{(N)}, \hat{x}_0^{(N)})$, where:
\begin{itemize}
\item $\hat{X}^{(N)}$ is a set of all vectors of the form $\langle \hat{X}_1^{(N)},\ldots,\hat{X}_n^{(N)}\rangle$, where for all $1\leq i\leq n$ it holds that $0\leq\hat{X}_i^{(N)}\leq 1$ and $\sum_{i}\hat{X}_i^{(N)}=1$. These vectors are the states of the system.
\item $\hat{x}_0^{(N)}\in \hat{X}^{(N)}$, is a vector of dimension $n$ which shows the initial state of the system. When the initial state of each component $\mathcal{C}_i$ is $s_{0_i}$, it is defined as:
\[\hat{x}_0^{(N)}=\langle \frac{\sum^{N}_{i=1} \textbf{1}\{s_{0_i}=s_1\}}{N},\ldots,\frac{\sum^{N}_{i=1} \textbf{1}\{s_{0_i}=s_n\}}{N} \rangle\]

where for $b$ a boolean formula:
\[\textbf{1}\{~b~\}=\left\{
  \begin{array}{l l}
    1\:: & \quad \text{if $b$ \textit{true}}\\
    0\:: & \quad \text{if $b$ \textit{false}}
  \end{array} \right.\]
is an indicator function.
\item $\mathcal{F}^{(N)}=\{f_1,\ldots,f_m\}$ is a set of transitions of the form $f_j=(a_j, s_j^{(N)}, t_j^{(N)}, r_j^{(N)})$ where:
\begin{itemize}
\item $a_j\in A$ is the label of the transition.
\item $s_j^{(N)}$ is a vector of dimension $n$, which shows which portion of each entry in a vector $\hat{x}\in \hat{X}^{(N)}$ is going to be used up by this transition.
\item $t_j^{(N)}$ is a vector of dimension $n$, which shows which values will be added to each entry in a vector $\hat{x}\in \hat{X}^{(N)}$ after this transition.
\item $r_j^{(N)}:A\times \hat{X}^{(N)}\rightarrow \mathbb{R}$ is the rate function of the transition.
\end{itemize}
Moreover, we define the \textit{state-change vector} $\nu_j^{(N)}=(t_j^{(N)}-s_j^{(N)})$ of $f_j$, showing the changes in the state of the system due to the transition $f_j$.
\end{itemize}

\end{definition}

Here we have not demonstrated how $\mathcal{F}^{(N)}$ is constructed, since it heavily depends on the way communications happen in a certain domain. In this paper we will give two possible procedures (Transformations \ref{prop1} and \ref{prop2}) which bridge the gap between component transition systems of wireless network nodes and the set of transitions $\mathcal{F}^{(N)}$ of a network of $N$ nodes.

In the following, the term PCTMC models always refers to normalized PCTMC models, unless otherwise stated. Next we give a formal definition of a System of Ordinary Differential Equations (ODEs) for a PCTMC model.

\begin{definition}
(System of ODEs). Let $\chi_{\mathcal{M}}^{(N)}=(\hat{X}^{(N)}, \mathcal{F}^{(N)}, \hat{x}_0^{(N)})$ be a PCTMC model. We consider an $n$-dimensional vector of functions $\hat{x}(t)=\langle \hat{x}_1(t),\ldots,\hat{x}_n(t)\rangle$ where for every $t \in \mathbb{R}_{\geq 0}$, $\hat{x}(t)\in \hat{X}^{(N)}$. The System of Ordinary Differential Equations for the PCTMC model $\chi_{\mathcal{M}}^{(N)}$ is defined as:
\begin{equation}\label{eqn:ode}
(\frac{d\hat{x}_1(t)}{dt},\ldots,\frac{d\hat{x}_n(t)}{dt})= \sum_{f\in \mathcal{F}} \nu_{f}^{(N)}\cdot r_{f}^{(N)}
\end{equation}
together with the initial condition:
\begin{equation}\label{eqn:ode2}
\langle\hat{x}_1(0),\ldots,\hat{x}_n(0)\rangle=\hat{x}^{(N)}_{0} 
\end{equation}
where we have $f=(a_f,s_f^{(N)},t_f^{(N)},r_f^{(N)})$.
\end{definition}

The system of ODEs above, together with the initial condition of the PCTMC model, can be solved to find a solution (or a number of solutions) which are $n$ curves showing the time evolution of the system. 
\iffalse
Therefore, the goal of this work is to define the set of transitions $\mathcal{F}$ in a way that first yields a PCTMC model, and then helps us derive an approximation of the System model which can be solved in practice despite the large number of components within the system.
\fi

In this paper we intend to use theorem 4.1 from \cite{bortolussi2013continuous} as a basis for the validity of the system of ODEs in approximating the system model. This requires the PCTMC model to be \textit{density-dependent}. Therefore we first present the concept of Lipschitz continuity.

\begin{definition}\label{def:lipschitz}
(Lipschitz Continuity). A function f is called Lipschitz continuous iff there is a positive, real constant $C$, such that for all points $x$ and $y$ over its domain:
\[|f(x)-f(y)|\leq C\cdot |x-y|.\]
\end{definition}

Note that any continuous function or any function with a bounded first derivative is Lipschitz continuous.

\begin{definition}\label{def:densitydep}
(Density-dependency). A PCTMC model $\hat{\chi}_{\mathcal{M}}^{(N)}=(\hat{X}^{(N)}, \mathcal{F}^{(N)}, \hat{x}_0^{(N)})$, is density-dependent iff the following conditions hold for all $\hat{x}\in\hat{X}^{(N)}$ and $f\in\mathcal{F}^{(N)}$:
\begin{enumerate}
\item for the state-change vector $\nu_f^{(N)}$ of $f$, the vector $N\cdot\nu_f^{(N)}$ is independent of $N$.
\item there is a Lipschitz continuous and bounded function $g_{f}:\hat{X}^{(N)}\rightarrow \mathbb{R}_{\geq 0}$ such that for all $\hat{x}$ the rate function $r_f^{(N)}$ scales with system size $N$ as:
\[r_f^{(N)}(\hat{x})=N\cdot g_{f}(\hat{x}).\]
\end{enumerate} 
\end{definition}

Now we present the main theorem (theorem 4.1 from \cite{bortolussi2013continuous}), which states that systems of ODEs are approximations for density-dependent PCTMC models when the number of components $N$ is sufficiently large. 

For every $N$, let $\hat{\mathcal{M}}^{(N)}=(S,Q)$ be continuous-time Markov chains for PCTMC models $\chi_{\mathcal{M}}^{(N)}=(\hat{X}^{(N)}, \mathcal{F}^{(N)}, \hat{x}_0^{(N)})$ where the set $S$ and the generator matrix $Q$ are defined as follows: 
\begin{itemize}
\item $S=\hat{X}^{(N)}$, meaning that  $\hat{\mathcal{M}}^{(N)}$ and $\chi_{\mathcal{M}}^{(N)}$ have the same state space.
\item For each state $\hat{x}^{(N)} \in \hat{X}^{(N)}$ and for each transition $f=(a, s^{(N)}, t^{(N)}, r^{(N)})$ in $\mathcal{F}^{(N)}$, if $\hat{x}^{(N)}$ allows $f$, meaning that in the vector $(\hat{x}^{(N)} - s^{(N)})$ there are no negative entries, and if $(\hat{x}^{(N)}+\nu_f^{(N)})\in \hat{X}^{(N)}$, then we have the following entry in the generator matrix $Q$:
\[q_{ij} =r^{(N)}~,~\text{with}~i=\hat{x}^{(N)}\text{and}~j=\hat{x}^{(N)}+\nu_f^{(N)}.\]
\end{itemize}

We define the transition probability matrix $P_{\hat{\mathcal{M}}^{(N)}}(t)$, with entries $p_{ij}(t)$ defining the probability of being in state $j$ at time $t$, if we have been in state $i$ at time $0$. Now we define the Markov process $\{\hat{\mathcal{M}}^{(N)}(t):t\geq 0\}$ where $\mathbb{P}\{\hat{\mathcal{M}}^{(N)}(0)=\hat{x}_0^{(N)}\}=1$, and:
\[\mathbb{P}\{\hat{\mathcal{M}}^{(N)}(t)=\hat{x}^{(N)} \}=p_{ij}(t)~,~\text{where}~i=\hat{x}_0^{(N)}~\text{and}~j=\hat{x}^{(N)}.\]

\begin{theorem}
(Deterministic Approximation for PCTMCs). For values $N\geq 1$ let $\chi_{\mathcal{M}}^{(N)}=(\hat{X}^{(N)}, \mathcal{F}^{(N)}, \hat{x}_0^{(N)})$ be density-dependent PCTMCs, and let $\{\hat{\mathcal{M}}^{(N)}(t): t\geq 0\}$ denote the corresponding Markov processes. Assume that for some point $\hat{x}_0\in \hat{X}^{(N)}$, we have $\lim_{N\rightarrow \infty}\hat{x}_0^{(N)}=\hat{x}_0$. Moreover, assume that the solution to the system of ODEs for $\chi_{\mathcal{M}}^{(N)}$ is $\hat{x}(t)$. Then for any finite time horizon $T<\infty$, it holds that:
\[\mathbb{P}\{\lim_{N\rightarrow \infty}\sup_{0\leq t\leq T}||\hat{\mathcal{M}}^{(N)}(t)-\hat{x}(t)||=0\}=1\]

\iffalse
\[\mathbb{P}\{\lim_{N\rightarrow \infty}\sup_{0\leq t\leq T}||\hat{\mathcal{M}}^{(N)}(t)-\hat{x}(t)||=0\}=1\]
\fi
\end{theorem}

The theorem above states that as $N$ grows, the difference between the behavior of the explicit model $\hat{\mathcal{M}}^{(N)}(t)$ and its approximation $\hat{x}(t)$ diminishes almost surely.

\section{Modeling interference in Slotted ALOHA networks}\label{sec:interference}
In this section we are first going to discuss interference in wireless networks and how it can be modeled and then apply this modeling approach to a Slotted ALOHA network.

\subsection{Interference}\label{subsec:interference}
In a WSN an increase in the number of nodes trying to send messages simultaneously increases the chances of interference between the signals. In this section, we are going to use results from \cite{van1990stability} to better express the notion of interference in networks. 

Despite interference, in wireless networks there is always a chance for one of the signals to be strong enough to \textit{capture} a receiver \cite{van1990stability}. We consider a set of $i$ senders within range of a receiver antenna, and we express the probability of a signal from one of the nodes capturing the receiver from a distance $r_{t}$ as:
\begin{equation}\label{eqn:capture}
q(i)=i\cdot\mathbb{P}\{\textit{capture}~|~\text{$(i-1)$ interfering signals}~,~\text{distance $r_t$}\}
\end{equation}

Since any of the $i$ messages has a chance of getting through, and these probabilities are mutually exclusive, meaning that it is not possible for two or more packets to capture the receiver simultaneously (therefore always $q(i)\leq 1$).

We express the scattering pattern of nodes around a receiver as a probability density function $f(r)$, henceforth called spatial distribution function, which gives the probability of a node being at distance $r$ from the receiver. The conditional probability of capture for a node becomes \cite{van1990stability}:
\begin{equation}\label{eqn:qi}
q(i)=i\cdot\int_{0}^{\infty}\bigg[\int_{0}^{\infty}\frac{f(r)}{1+zr_t^{\beta}r^{-\beta}}dr\bigg]^{(i-1)}f(r_t) dr_t
\end{equation}

Parameter $z$ describes a power threshold: when facing $i-1$ interfering signals, the power of a successful message should always exceed the power of the joint $i-1$ interfering signals by $z$. Parameter $\beta$ shows pathloss: the signal gets weaker as distance r increases, according to $r^{-\beta}$. The term in brackets is the probability of surviving a single interfering signal when the power threshold is $z$ and the sender is at distance $r_t$, and the enclosing integral averages over different values for $r_t$ when we know that it is distributed according to spatial distribution function $f(r)$.

We assume that the spatial distribution function $f(r)$ can have two patterns. A uniform $f(r)$ is:
\begin{equation}\label{eqn:uniform}
f(r)=\left\{
        \begin{array}{l l}
        2r&\text{for}~0\leq r\leq1,\\
        0&\text{for}~r>1.\\
        \end{array}
    \right.
\end{equation}
Which is the probability of the node being on a circle with radius $r$ around the receiver, where $r$ is taken to be a normalized distance. A log-normal $f(r)$ is (\cite{van1990stability}):
\begin{equation}\label{eqn:lognormal}
f(r)=\frac{\beta}{\sqrt{2\pi}r\sigma_d}\exp{\Big\{-\frac{\beta^2 \log{r^2}}{2\sigma_d^2}\Big\}}.
\end{equation}

Here, $\sigma_d^2$ is the spatial logarithmic variance. The log-normal distribution is more realistic, in the sense that no node can be infinitely close to the receiver and hence the probability of capture when the number of nodes is very high will actually go to zero, which is the case in reality.

Plots of the uniform and log-normal distributions in formula (\ref{eqn:uniform}) and formula (\ref{eqn:lognormal}) for $\beta=4$ and $\sigma_d=2$ are presented in Fig.~\ref{fig:plots}.
\begin{figure}[h]
\centering
\includegraphics[scale=.8]{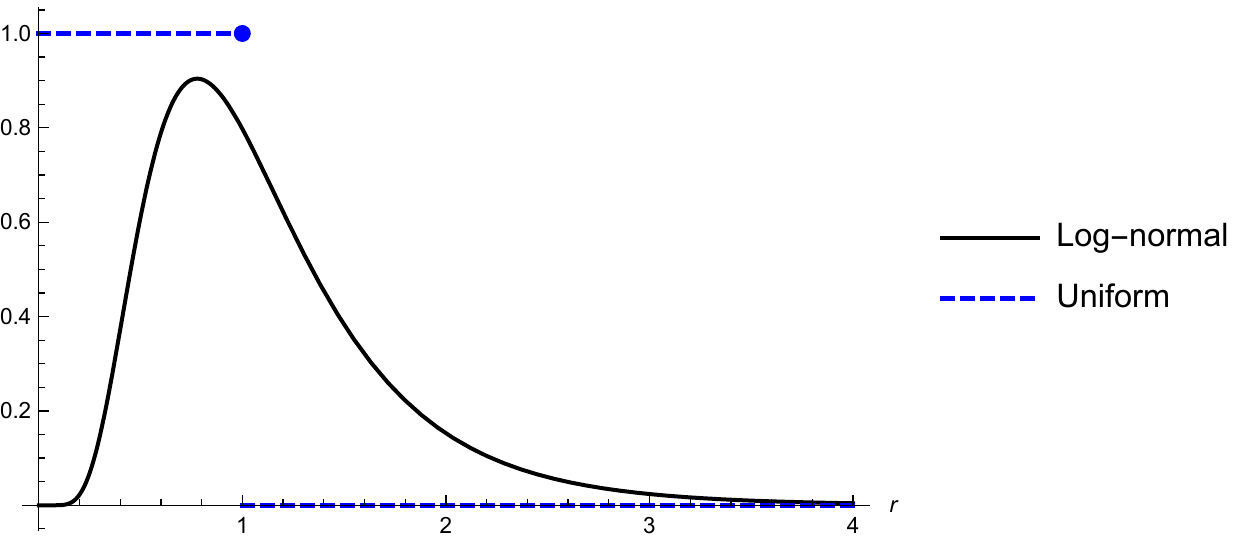}
\caption{Plot of the uniform and log-normal spatial distributions} \label{fig:plots}
\end{figure}

\begin{lemma}\label{lem:lem1}
Let $f(r)$ be any spatial distribution for which $\lim_{r\to \infty} f(r)=0$ and also $\lim_{r\to \infty} \frac{f(r)^2}{f'(r)}=0$. For bounded and positive values of parameter $z$, the function $q(i)$ in formula (\ref{eqn:qi}) is Lipschitz continuous over domain $[1,\infty)$.
\end{lemma}

\begin{proof}
Taking the first derivative of \textit{formula \ref{eqn:qi}} yields:
\begin{flalign*}
\frac{dq}{d i}=&\int_{0}^{\infty}\bigg[\int_{0}^{\infty}\frac{f(r)}{1+zr_t^{\beta}r^{-\beta}}dr\bigg]^{(i-1)}f(r_t) dr_t\\
&+i\cdot\big(\int_{0}^{\infty}\bigg[\int_{0}^{\infty}\frac{f(r)}{1+zr_t^{\beta}r^{-\beta}}dr\bigg]^{(i-1)}\cdot \log{\big(\int_{0}^{\infty}\frac{f(r)}{1+zr_t^{\beta}r^{-\beta}}dr\big)}\cdot f(r_t) dr_t\big)\\
\end{flalign*}
which we show to be bounded over the domain $[1,\infty)$. We establish that the two addends are bounded, and hence the entire term is bounded.

\begin{itemize}
\item The first addend: the assumption is that $z$ is positive, therefore for any positive $r$ and $r_t$: 
\[1+zr_t^{\beta}r^{-\beta}>1 ~~~~~ \Longrightarrow ~~~~~\frac{1}{1+zr_t^{\beta}r^{-\beta}}<1\]

Since $f(r)$ is a probability density function, we have $\int_{0}^{\infty}f(r)dr=1$. Therefore, for any $r_t$:
\[\int_{0}^{\infty}\frac{f(r)}{1+zr_t^{\beta}r^{-\beta}}dr<1~~~~~~~~~~~~~~~~~~~~~~~~~~~~~~~~~~~~~~~~\]
\[~~~~~ \Longrightarrow ~~~~~\int_{0}^{\infty}\bigg[\int_{0}^{\infty}\frac{f(r)}{1+zr_t^{\beta}r^{-\beta}}dr\bigg]^{(i-1)}f(r_t) dr_t<\int_{0}^{\infty}[1]^{(i-1)}f(r_t) dr_t=1\]
and the first addend is bounded.

\item The second addend: the new term $\log{\big(\int_{0}^{\infty}\frac{f(r)}{1+zr_t^{\beta}r^{-\beta}}dr\big)}$ is unbounded (is -$\infty$) when:
\[\int_{0}^{\infty}\frac{f(r)}{1+zr_t^{\beta}r^{-\beta}}dr=0\]

But according to assumptions, since $\int_{0}^{\infty}f(r)dr=1$ and $\lim_{r\rightarrow\infty}f(r)=0$ for some finite $r$, $f(r)>0$. Therefore $-\infty$ happens only when $r_t\rightarrow +\infty$. In that case it is sufficient to show that:
\[\lim_{r_t\to +\infty} \log{\big(\int_{0}^{\infty}\frac{f(r)}{1+zr_t^{\beta}r^{-\beta}}dr\big)} \cdot f(r_t)\]

is zero. We know that: 
\[\lim_{r_t\to +\infty} \frac{1}{f(r_t)}=+\infty ~~ \Longrightarrow ~~  \lim_{r_t\to +\infty} \frac{\log{\big(\int_{0}^{\infty}\frac{f(r)}{1+zr_t^{\beta}r^{-\beta}}dr\big)}}{\frac{1}{f(r_t)}}=\frac{-\infty}{+\infty}\]
Therefore, there is sufficient ground to use the \textit{L'H\^{o}pital's rule}:
\[\lim_{r_t\to +\infty} \frac{\log{\big(\int_{0}^{\infty}\frac{f(r)}{1+zr_t^{\beta}r^{-\beta}}dr\big)}}{\frac{1}{f(r_t)}} \Heq \lim_{r_t\to +\infty}\frac{\int_{0}^{\infty}\frac{-\beta f(r)}{r^{-\beta}}dr}{\int_{0}^{\infty}\frac{f(r)}{zr^{-\beta}}dr}\cdot\frac{f(r_t)^2}{f'(r_t)}\]

which is zero since according to the assumptions:
\[\lim_{r_t\to +\infty} \frac{f(r_t)^2}{f(r_t)}=0\]
and the second addend is also bounded.
\end{itemize}
\end{proof}

\begin{corollary}\label{cor:cor1}
Substituting the $f(r)$ functions in formulae (\ref{eqn:uniform}) and (\ref{eqn:lognormal}) in the definition of $q(i)$ in (\ref{eqn:qi}) yields Lipschitz continuous functions.
\end{corollary}

The interested reader with common skills in calculus could check that the log-normal spatial distribution in formula (\ref{eqn:lognormal}) satisfies the equation:
\[\lim_{r\to \infty} \frac{f(r)^2}{f'(r)}=0.\]
and is hence Lipschitz continuous. As for the uniform spatial distribution in formula (\ref{eqn:uniform}), matters are not as easy, since:
\[\lim_{r\to \infty} \frac{f(r)^2}{f'(r)}=\frac{0}{0}=\text{undefined}.\]

For the uniform spatial distribution, and for the specific case when $\beta=4$, we propose the following lemma.

\begin{lemma}\label{lem:lem11}
The function:
\[q(i)=i\cdot\int_{0}^{1}\bigg[\int_{0}^{1}\frac{2r}{1+zr_t^{4}r^{-4}}dr\bigg]^{(i-1)}2r_t dr_t\]
\[~~~~~~~~~~~~~=i\cdot\int_{0}^{1}\bigg[1-\sqrt{z}r_t^2\arctan(\frac{1}{\sqrt{z}r_t^2})\bigg]^{(i-1)}2r_t dr_t
\]
is Lipschitz continuous.
\end{lemma}

\begin{proof}
The first derivative of $q(i)$ is:
\begin{flalign*}
q'(i)=&\int_{0}^{1}\bigg[1-\sqrt{z}r_t^2\arctan(\frac{1}{\sqrt{z}r_t^2})\bigg]^{(i-1)}2r_t dr_t\\
&+i\cdot\big(\int_{0}^{1}\bigg[1-\sqrt{z}r_t^2\arctan(\frac{1}{\sqrt{z}r_t^2})\bigg]^{(i-1)}\cdot \log{\big(1-\sqrt{z}r_t^2\arctan(\frac{1}{\sqrt{z}r_t^2})\big)}2r_t dr_t\big)\\
\end{flalign*}

For the first line we know $0\leq\arctan(x)\leq \frac{\pi}{2}~,~x\in\mathbb{R}_{\geq 0}$;
therefore for $i\geq 1$ the first factor is bounded, and by integrating over the bounded domain $[0,1]$, it will still be bounded in value.

For the second line, we again have the same term multiplied by:
\[\log{\big(1-\sqrt{z}r_t^2\arctan(\frac{1}{\sqrt{z}r_t^2})\big)}\]
which is independent of $i$. This term is unbounded (evaluates to $-\infty$) only when:
\[\big(1-\sqrt{z}r_t^2\arctan(\frac{1}{\sqrt{z}r_t^2})\big)=0\] 
\[\Longrightarrow ~~\arctan(\frac{1}{\sqrt{z}r_t^2})=\frac{1}{\sqrt{z}r_t^2}\]
\[\Longrightarrow~~~~~~~~~~~~~~~~~~~~~~~\frac{1}{\sqrt{z}r_t^2}=0.\]
But this is not the case here since $z$ is bounded and $0<r_t<1$. After integrating over the bounded domain $[0,1]$ and multiplication by $i$ the term is still bounded.
\end{proof}

In Fig.~\ref{fig:qi}, values numerically derived for formula (\ref{eqn:qi}) are shown when the number of signals is between 1 to 10, for $\beta=4$, $z=10$ and $\sigma_d=2$. For uniformly distributed nodes (the points are shown with `+'), the curve will stabilize over a fixed value ($\sim 0.2$) for large numbers. But when using log-normal spatial distribution to calculate the capture probability $q(i)$ (the points are shown with `$\circ$'), the curve goes to zero.

\begin{figure}[h]
\centering
\scalebox{0.85}{
\begin{tikzpicture}[y=.2cm, x=.7cm,font=\sffamily]
 	%axis
	\draw (0,0) -- coordinate (x axis mid) (10,0);
    	\draw (0,0) -- coordinate (y axis mid) (0,20);
    	%ticks
    	\foreach \x in {0,...,10}
     		\draw (\x,1pt) -- (\x,-3pt)
			node[anchor=north] {\x};
    	\foreach \y in {0,0.2,0.4,0.6,0.8,1}
     		\draw (1pt,\y*20) -- (-3pt,\y*20) 
     			node[anchor=east] {\y}; 
	%labels      
	\node[below=0.8cm] at (x axis mid) {Number of nodes (i)};
	\node[rotate=90, above left=1.2cm] at (y axis mid) {Probability $q$};

	%plots
	\draw plot[mark=*, mark options={fill=white}] 
		file {mortal.data};
	\draw plot[mark=+, mark options={fill=black} ] 
		file {immortal.data};

\end{tikzpicture}
}
\caption{The probability $q(i)$ for $0\leq i\leq 10$. The two set of points show two different patterns: $q(i)$ with log-normal spatial distribution \textup{($\circ$)} tends to $0$ for very large $i$, while with uniform distribution \textup{(+)} it converges to a fixed value ($\sim 0.2$). The lines show a possible interpolating function as a continuous approximation.} \label{fig:qi}
\end{figure}
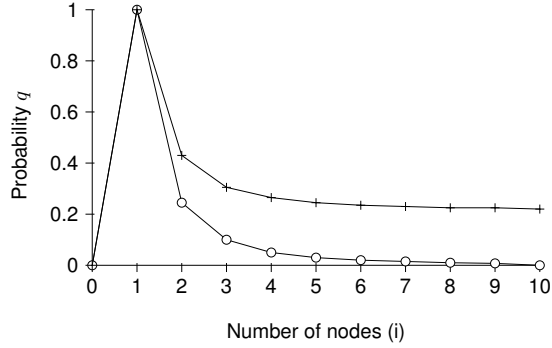

In the rest of this paper, function $q(i)$ is used over the domain $[0,\infty)$. So we assume that the value of $q(i)$ in the range $[0,1)$ is defined as: $q(i)=i$, since we know that for a single transmitting node (or less) there is no interference in the network. This also assures continuity, since $q(1)=1$ in both functions.

The following summarizes the results of this section in terms of the definition of the set $\mathcal{F}$ in PCTMC models for WSNs.

\begin{transformation}\label{prop1}
For component $\textit{Node}=(S, Act, \rightarrow, s_0)$, assume that the transitions in the set $\rightarrow$ are restricted to the following forms:
\begin{itemize}
\item $s_{i}\xrightarrow{(\text{capture}, r_{send})} s_{j}\in \rightarrow$ be ``successful send'' transitions,
\item $s_{i}\xrightarrow{(\text{failure}, r_{send})} s_{k}\in \rightarrow$ be ``failed send'' transitions, due to interference,
\item $s_{l}\xrightarrow{(\alpha, r_{\alpha})} s_{m}\in \rightarrow$ be any other transition.
\end{itemize}

Let $\textit{Network}\pepaeq(\textit{Node}_{1}^{(N)}~||\ldots||~\textit{Node}_{N}^{(N)})$ be a system made of $N$ identical components of type \textit{Node}. The set $\mathcal{F}^{(N)}$ of the associated PCTMC model $\hat{\chi}_{\textit{Network}}^{(N)}=(\hat{X}^{(N)}, \mathcal{F}^{(N)}, \hat{X}_0^{(N)})$ defined for $\hat{x}$, where $\hat{x} \in \hat{X}^{(N)}$ consists of the following tuples:
\begin{flalign*}
f_{\text{failure}_i}^{(N)}&=(\text{failure},~\frac{e_i}{N},~\frac{e_j}{N},~p_{send}\cdot (N\cdot\hat{x}_i-q(N\cdot\hat{x}_i)))&\\
f_{\text{capture}_i}^{(N)}&=(\text{capture},~\frac{e_i}{N},~\frac{e_k}{N},~p_{send}\cdot q(N\cdot\hat{x}_i))&\\
f_{\alpha_l}^{(N)}~~~~&=(\alpha,~\frac{e_l}{N},~\frac{e_m}{N},~r_{\alpha}\cdot N\cdot\hat{x}_l)&
\end{flalign*}
where $e_{i}$ is a unit vector with value $1$ at position $i$.
\end{transformation}

\begin{lemma}\label{lem:lem2}
Let $N\geq 1$ and function $q(i)$ be a Lipschitz continuous function over domain $i \in (0,\infty)$ and right-continuous at point $i=0$. The PCTMC model $\hat{\chi}_{\textit{Network}}^{(N)}=(\hat{X}^{(N)}, \mathcal{F}^{(N)}, \hat{x}_0^{(N)})$ is density-dependent.
\end{lemma}

\begin{proof} 
In order to prove density-dependency for $\hat{\chi}_{\textit{Network}}^{(N)}$ we check the conditions in Definition \ref{def:densitydep}:
\begin{enumerate}
\item The vectors:
\begin{flalign*}
N\cdot \nu_{f_{\text{capture}_i}^{(N)}}&= N\cdot (\frac{e_j}{N}-\frac{e_i}{N})\\
&=(e_j-e_i)~\text{,}\\
N\cdot \nu_{f_{\text{failure}_i}^{(N)}}&= N\cdot (\frac{e_k}{N}-\frac{e_i}{N})\\
&=(e_k-e_i)~\text{,}\\
N\cdot \nu_{f_{\alpha_l}^{(N)}}~~~&=N\cdot(\frac{e_m}{N}-\frac{e_l}{N})\\
&=(e_m-e_l)\\
\end{flalign*}
are independent of $N$.

\item As for the continuity criteria, we consider the $g_f$ functions below to be functions of the entries of the vector $\hat{X}^{(N)}$ and prove their Lipschitz continuity:
\begin{itemize}
\item[-] The function $g_{f_{\alpha_l}^{(N)}}=r_{\alpha} \cdot\hat{x}_l$ has the first derivative $\frac{\partial g_{f_{\alpha_l}^{(N)}}}{\partial \hat{x}_l}=r_{\alpha}$, where $r_\alpha$ is constant and therefore continuous and bounded. This means that $g_{f_{\alpha_l}^{(N)}}$ is Lipschitz continuous.

\item[-] For the function $g_{f_{\text{failure}_i}^{(N)}}=r_{send} \cdot\frac{1}{N}\cdot q(N\cdot\hat{x}_i)$,  
the term $r_{send} \cdot\frac{1}{N}$ would be a constant for a predetermined value of $N\geq 1$. As for the second part, we know that $0\leq\hat{x}_i\leq 1$. But then $0\leq N\cdot\hat{x}_i\leq N$ since $N\geq 1$ according to the assumptions. Again, based on the assumptions we know that $q(i)$ is Lipschitz continuous over domain $i\in[0,\infty)$. Therefore the term $g_{f_{\text{failure}_i}^{(N)}}$ is Lipschitz continuous.

\item[-] The function $g_{f_{\text{capture}_i}^{(N)}}=r_{send}\cdot (\hat{x}_i-\frac{1}{N}\cdot q(N\cdot\hat{x}_i))$, has two factors, the first of which is \[r_{send}\cdot \hat{x}_i\]
and is Lipschitz continuous like the first case.
And the second term 
\[\frac{r_{send}}{N}\cdot q(N\cdot\hat{x}_i)\]
is Lipschitz continuous following the second case.

\end{itemize}
\end{enumerate}
\end{proof}

\subsection{Slotted ALOHA with a single receiver}\label{subsec:simple}
In this part, we will consider a PCTMC model, built according to Transformation~\ref{prop1}, where every Node component in the system runs the Slotted ALOHA protocol. We consider the same scenario as in \cite{van1990stability} which consists of a number of senders scattered around a single antenna. The aim is to use the mean field approximation method to observe the bistable behavior of this specific ALOHA network.

The Slotted ALOHA protocol we consider here is expressed by the following set of rules \cite{abramson1970aloha,roberts1975aloha}:
\begin{itemize}
\item Whenever there is data to send, send it at the start of the next time-slot.
\item If the message could not be delivered due to interference, retry sending the message.
\end{itemize}
To this we also add the following restriction:
\begin{itemize}
\item While sending and retrying, do not generate new messages.
\end{itemize}

A node's behavior is presented in Fig.~\ref{fig:simple}. A node does its internal processing in state ($O$), and generates a new message with rate $r_o$. Next, the node enters a state where it transmits the message ($T$). While sending with rate $r_{send}$, if other nodes are also transmitting messages simultaneously, the signals will interfere. In case the message cannot be delivered, a node enters the backlog state ($R$), where it tries to retransmit the message after some time ($r_r$).

\iffalse
Here, we have also made the assumption that the timeslots are synchronized between all nodes, however there is little difference when we do not have synchronous transmissions.
\fi

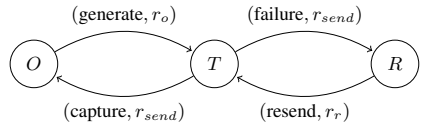
\begin{figure}[h]
\centering
\scalebox{0.8}{
\begin{tikzpicture}[shorten >=1pt,node distance=3cm,on grid,auto]
   \node[state] (q_0)   {$O$};
   \node[state] (q_1) [right=of q_0] {$T$};
   \node[state] (q_2) [right=of q_1] {$R$};
    \path[->]
    (q_0) edge [bend left] node {$(\text{generate}, r_{o})$} (q_1)
    (q_1) edge [bend left] node  {$(\text{capture}, r_{send})$} (q_0)
          edge [bend left]  node {$(\text{failure}, r_{send})$} (q_2)
    (q_2) edge [bend left]  node {$(\text{resend}, r_{r})$} (q_1);
\end{tikzpicture}
}
\caption{An ALOHA node's component transition system}
\label{fig:simple}
\end{figure}

Now we proceed to generate the set of differential equations associated with the slotted ALOHA network. Following the theory presented in \textit{section \ref{sec:method}}, the number of processes in state $s\in\{O,T,R\}$ at time $t$ is expressed by the function $\hat{x}_s(t)$, and following Transformation~\ref{prop1} we derive the set of equations in Table~\ref{tbl:simplealohanodes}.

\begin{table}[h]
\caption{System of ODEs for the Slotted ALOHA model}\label{tbl:simplealohanodes}
\fbox{
\begin{minipage}{15em}
\begin{flalign*}
\frac{d\hat{x}_{O}(t)}{dt}&=-r_o\cdot\hat{x}_{O}(t) + \frac{1}{N}\cdot r_{send}\cdot q(N\cdot\hat{x}_{T}(t)) &\\
\frac{d\hat{x}_{T}(t)}{dt}&=-r_{send}\cdot (\hat{x}_{T}(t)-\frac{1}{N}\cdot q(N\cdot\hat{x}_{T}(t))) - \frac{1}{N}\cdot r_{send}\cdot q(N\cdot\hat{x}_{T}(t)) &\\&+ r_o\cdot\hat{x}_{O}(t) + r_r\cdot\hat{x}_{R}(t) &\\
\frac{d\hat{x}_{R}(t)}{dt}&=-r_r\cdot\hat{x}_{R}(t) + r_{send}\cdot (\hat{x}_{T}(t)-\frac{1}{N}\cdot q(N\cdot \hat{x}_{T}(t)))
\end{flalign*}
\end{minipage}
}
\end{table}

Following \cite{van1990stability} we take $r_r=0.08$, $r_o=0.0055$ and $r_{send}=1$; since a message transmission always takes one timeslot. Specifying the initial condition as ($\hat{x}_O(0) = 1$, $\hat{x}_T(0) = 0$, $\hat{x}_R(0) = 0$), and solving the equations numerically (because the complexity of function $q(i)$ does not allow explicit solutions) we have the curves in Fig.~\ref{fig:na}. Here the curves show the number of nodes in each state, which change over time. After some changes, they stabilize over an \textit{equilibrium point}, which is the \textit{fixpoint} of the system of ODEs. The plot shows a network with a good behavior, in which once in a while a new message is generated and is almost always successfully delivered to the receiver in the first try ($\hat{x}_{R}(t)\sim 0$).

However, under different initial conditions, namely ($\hat{x}_O(0) = 0$, $\hat{x}_T(0) = 1$, $\hat{x}_R(0) = 0$), the solution is the curves in Fig.~\ref{fig:nb}. This plot shows that if the system starts in a state where everybody is trying to transmit a message, then it will be trapped in a state with a very low throughput where a constant number of nodes are always trying to deliver their messages and saturate the media.

In Fig.~\ref{fig:nc}, we see that in a less realistic case of traffic in a uniform spatially distributed network, the nodes tend to operate efficiently after some time despite the initial conditions.

\begin{figure}
\centering
\begin{subfigure}[b]{0.45\textwidth}
\includegraphics[scale=0.43]{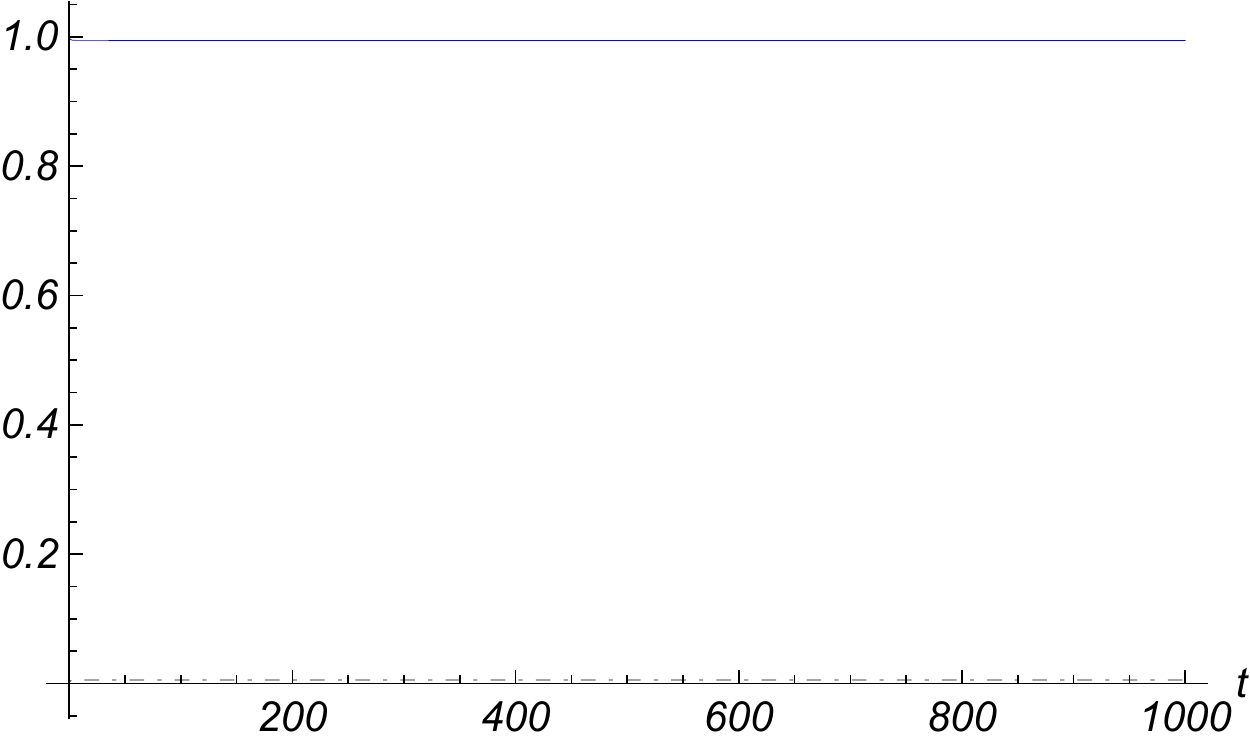}
\caption{The solution, with log-normal spatial distr. and with $\hat{x}_O(0) = 1$} \label{fig:na}
\end{subfigure}
\begin{subfigure}[b]{0.45\textwidth}
\includegraphics[scale=0.6]{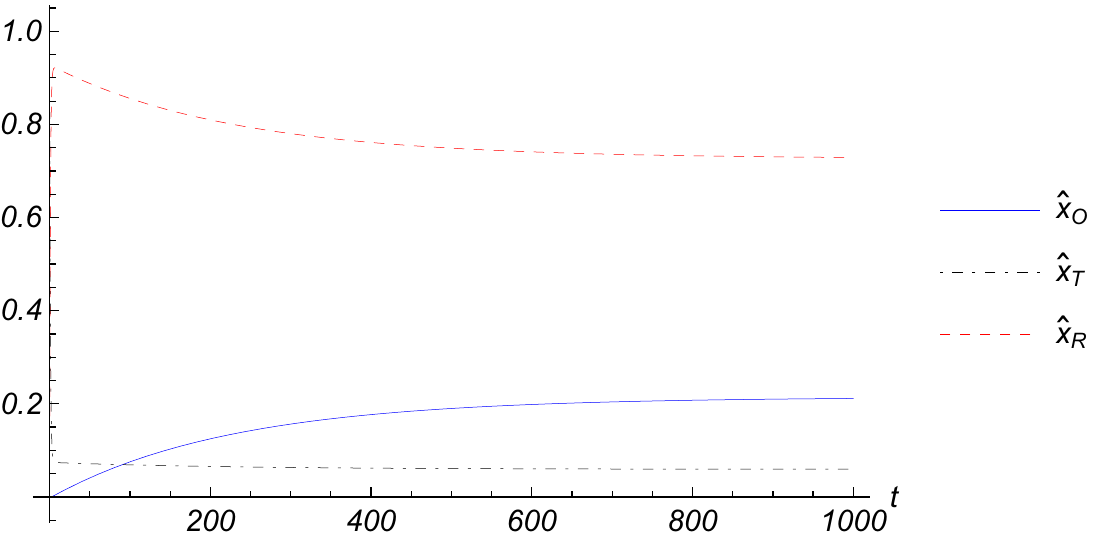}
\caption{The solution, with log-normal spatial distr. and with $\hat{x}_T(0) = 1$} \label{fig:nb}
\end{subfigure}
\begin{subfigure}[b]{0.45\textwidth}
\includegraphics[scale=0.6]{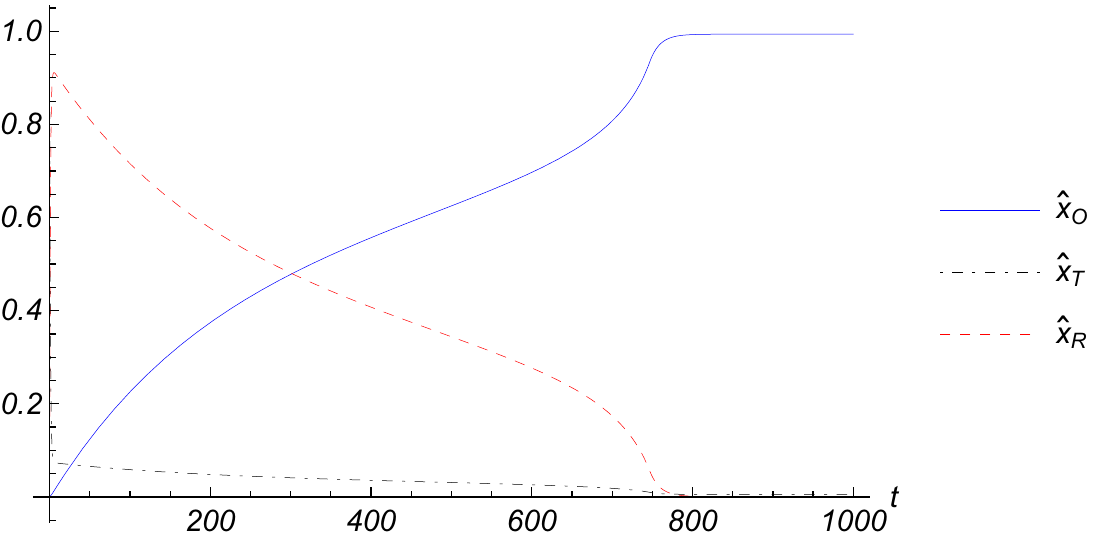}
\caption{The solution, with uniform spatial distr. and with $\hat{x}_T(0) = 1$} \label{fig:nc}
\end{subfigure}
\caption{Numerical solutions for the slotted ALOHA network} \label{fig:numericalresults}
\end{figure}

The bistability that we see in case of the log-normal spatial distribution is a common property of many ALOHA-type networks. The results match those presented in \cite{van1990stability}, and other observations on how ALOHA networks behave in real world \cite{rosenkrantz1983instability}.

\section{Modeling local broadcast in Slotted ALOHA networks}\label{sec:broadcast}
In this section we are first going to discuss local broadcast in wireless networks and how it can be modeled, and then apply the modeling approach to a network running a simple neighborhood discovery protocol.

\subsection{Local broadcast}\label{subsec:broadcast}
In essence, wireless network communication traffic consists of a number of broadcasts, in which every node in the vicinity of a sender is capable of hearing the message; so a message transmission typically involves one sender and several receivers.

The number of receiver parties depends on the technology of the radio modules, the properties of the media, the network topology, etc. Here, we consider a single measure $p$ to represent all these properties. The probability $p$ is the fraction of nodes that are close enough to the transmitter to be able to successfully receive the message, and usually has the form $p=\frac{d}{N}$, where $d$ is some constant number of nodes which are located in a certain neighborhood.

For this communication, two terms appear in ODEs; one for the receivers and one for senders. For the receiver, we start with the term used in \cite{bradley2008analysing} to describe communication when there is only one receiver. We call this $M_{\alpha}$, which has the general form:
\begin{equation}\label{eqn:singlereceive}
M_{\alpha}=r_{\alpha}\cdot N_{s}\cdot \textbf{1}\{N_r>0\}.
\end{equation}

Where $r_\alpha$ is the rate of send action $\alpha$, $N_{s}$ stands for the total number of processes in states in which a send action $\alpha$ is possible, $N_r$ is the total number of processes in states in which a receive action is possible, and the indicator function $\textbf{1}\{N_r>0\}$ shows the fact that communication is not possible when there are no receivers.

\iffalse
However, because the function is discontinuous, introducing it into differential equations will cause problems regarding the existence and uniqueness of a solution to ODEs.
\fi

Instead, in broadcasting we reintroduce the term $M_{\alpha}$ in (\ref{eqn:singlereceive}) as term $\Pi_{\alpha}$:
\begin{equation}\label{eqn:receive}
\Pi_{\alpha}=r_{\alpha}\cdot N_{s}\cdot (p\cdot N_{r}).
\end{equation}
since at each point in time we know that a portion of all receivers that are within range (expressed by parameter $p$) are capable of receiving the message. 

$\Pi_{\alpha}$ in equation (\ref{eqn:receive}) can be derived by rewriting the term $M_{\alpha}=r_{\alpha}\cdot N_{s}\cdot \textbf{1}\{N_r>0\}$ in a setting with multiple receivers. In order to show this, we first describe $\mathbb{P}\{M=i|N_r\}$ as the probability of $i$ successful receives, when there are a total of $N_r$ receivers:
\begin{equation}\label{eq:pmnr}
\mathbb{P}\{M=i|N_{r}\}=\dbinom{N_{r}}{i}p^i (1-p)^{N_{r}-i}.
\end{equation}
Because we know that any subset of the possible receivers are capable of participating in the wireless communication, we extend term (\ref{eqn:singlereceive}) and derive:
\[
\Pi_{\alpha}=\sum_{i=1}^{N} \mathbb{P}\{M=i|N_{r}\}\cdot i\cdot r_{\alpha}\cdot N_{s}\cdot \textbf{1}\{N_r\geq i\}
.\]

Where $N$ is the total number of nodes in the network. Replacing the term $\mathbb{P}\{M=i|N_{r}\}$ with the right hand side of (\ref{eq:pmnr}) we have:
\begin{align*}
\Pi_{\alpha}&=~\sum_{i=0}^{N_{r}}\big(\dbinom{N_{r}}{i}p^i (1-p)^{N_{r}-i}\big)\cdot i\cdot r_{\alpha}\cdot N_{s}\\
&=~r_{\alpha}\cdot N_{s}\cdot \big(\sum_{i=0}^{N_{r}}i \dbinom{N_{r}}{i}p^i (1-p)^{N_{r}-i}\big).\\
\end{align*}

The series is the expected value of a binomial distribution; therefore the receive term is simply:
\[\Pi_{\alpha}=r_{\alpha}\cdot N_{s}\cdot (p\cdot N_{r}).\]

Next, we see how results presented regarding interfering signals can be combined with local broadcasting and help us derive the term $R_{\alpha}$. Term (\ref{eqn:receive}) can be interpreted differently, as a portion of senders ($p\cdot N_{r}$) which are within range of each receiver, and try to capture it:
\[r_{\alpha}\cdot (p\cdot N_{s})\cdot N_{r}.\]

And then we apply the function $q(i)$ to the total number of received messages at each receiver's site:
\begin{equation}\label{eqn:receivecol}
R_{\alpha}=r_{\alpha}\cdot q(p\cdot N_{s})\cdot N_{r}.
\end{equation}
This is the term we intend to use for receivers. As for the senders, since a node which is sending a message type $\alpha$ does not depend on the status of receivers, we have the following simple format:
\begin{equation}\label{eqn:send}
T_{\alpha}=r_{\alpha}\cdot N_{s}.
\end{equation}
This means that when transmitting, the sender does not depend on the number of receivers in its vicinity. 

In the following, all send or receive actions have the form $send(m)$ or $receive(m)$ and are accompanied by a type $m\in M$, where $M$ is a set of message types. We define the set $I_{m}$ to contain send actions that are interfering with send actions with message type $m$, where at least $send(m)\in I_{m}$.

\begin{transformation}\label{prop2}
For component $\textit{Node}=(S, Act, \rightarrow, s_0)$, assume that the transitions in the set $\rightarrow$ are restricted to the following forms:
\begin{itemize}
\item $s_{i}\xrightarrow{(\text{send($m$)}, r_{\text{send}})} s_{j}\in \rightarrow$ be send transitions of message type $m$,
\item $s_{k}\xrightarrow{(\text{receive($m$)},\bot)} s_{l}\in \rightarrow$ be associated receive transitions of message type $m$,
\item $s_{m}\xrightarrow{(\alpha, r_{\alpha})} s_{o}\in \rightarrow$ be any other transition.
\end{itemize}
We forbid the component Node to be able to both send and receive a message type $m$ when in a state $s_i\in S$.

Let $\textit{Network}\pepaeq(\textit{Node}_{1}^{(N)}~||\ldots||~\textit{Node}_{N}^{(N)})$ be a system made of $N$ identical \textit{Node} components. Let $\hat{\chi}_{\textit{Network}}^{(N)}=(\hat{X}^{(N)}, \mathcal{F}^{(N)}, \hat{x}_0^{(N)})$ be the associated PCTMC model.

Consider $C(A)$ to be the portion of nodes that are in a state in which they are capable of doing the actions of set $A$, where for $\hat{x}^{(N)}_i \in \hat{X}^{(N)}$, it is defined as follows:
\begin{flalign*}
C(A)=\sum_{i \in \{a|\exists s_{b}.(s_{a}\xrightarrow{(\text{$\alpha$}, r_{\alpha})} s_{b}\wedge \alpha \in A)\}}\hat{x}_{i}^{(N)}
\end{flalign*}
The set $\mathcal{F}^{(N)}$ consists of the following tuples (again for $\hat{x}$ defined above):
\begin{align*}
f_{\text{send($m$)}_i}^{(N)}~~&=(\text{send($m$)},~\frac{e_i}{N},~\frac{e_j}{N},~N\cdot r_{send}\cdot \hat{x}_i)&\\
f_{\text{receive($m$)}_k}^{(N)}&=(\text{receive($m$)},~\frac{e_k}{N},~\frac{e_l}{N}, &\\
&~~~~~N\cdot r_{send}\cdot \frac{C(\{send(m)\})}{C(I_{m})}\cdot q(N\cdot p\cdot C(\{send(m)\}))\cdot\hat{x}_k) &\\ 
f_{\alpha_m}^{(N)}~~~~~~~&=(\alpha,~\frac{e_m}{N},~\frac{e_o}{N},~N\cdot r_{\alpha}\cdot\hat{x}_m)&
\end{align*}
where $e_{i}$ is a unit vector with value $1$ at position $i$.
\end{transformation}

It is worth noting that the term $N\cdot p=N\cdot \frac{d}{N}=d$, and therefore $f_{\text{receive($m$)}_k}^{(N)}$ does not depend on $N$.

\begin{lemma}\label{lem:lem3}
Let $N\geq 1$ and function $q(i)$ to be a Lipschitz continuous function over domain $i \in (0,\infty)$ and right-continuous at point $i=0$. The PCTMC model $\hat{\chi}_{\textit{Network}}^{(N)}=(\hat{X}^{(N)}, \mathcal{F}^{(N)}, \hat{x}_0^{(N)})$ is density-dependent.
\end{lemma}

\begin{proof}
We check the conditions in Definition \ref{def:densitydep}:

\begin{enumerate}
\item The vectors:
\begin{flalign*}
N\cdot\nu_{f_{\text{send($m$)}_i}^{(N)}}~~&= N\cdot (\frac{e_j}{N}-\frac{e_i}{N})\\
&=e_j-e_i~\text{,}\\
N\cdot\nu_{f_{\text{receive($m$)}_k}^{(N)}}&= N\cdot (\frac{e_l}{N}-\frac{e_k}{N})\\
&=e_l-e_k~\text{,}\\
N\cdot\nu_{f_{\alpha_m}^{(N)}}~~~~~~~&=N\cdot (\frac{e_o}{N}-\frac{e_m}{N})\\
&=e_o-e_m\\
\end{flalign*}
are independent of $N$.

\item As for the continuity criteria, we consider the $g_f$ functions below to be functions over vectors $\hat{x}\in\hat{X}^{(N)}$ and prove their Lipschitz continuity:
\begin{itemize}
\item[-] The function $g_{f_{\alpha_m}^{(N)}}=r_{\alpha} \cdot\hat{X}^{(N)}_m$ has the first derivative $\frac{\partial g_{f_{\alpha_m}^{(N)}}}{\partial \hat{X}^{(N)}_m}=r_{\alpha}$, where $r_\alpha$ is continuous and bounded. Therefore $g_{f_{\alpha_m}^{(N)}}$ is Lipschitz continuous.

\item[-] For the function 
\[g_{f_{\text{receive}_k}^{(N)}}=r_{send}\cdot \frac{C(\{send(m)\})}{C(I_{m})}\cdot q(N\cdot p\cdot C(\{send(m)\}))\cdot\hat{x}_i\]

We first make an observation. We know that 
\[C(A)=\sum_{i\in B}\hat{x}_{i}\]
where $B\subseteq \{1,\ldots,n\}$. Since $send(m)\in I_m$ we always have: 
\[C(\{send(m)\})\leq C(I_m)\]

Next, we show that for two $n$-dimensional vectors $\overrightarrow{x}$ and $\overrightarrow{y}$ in $\hat{X}^{(N)}$, there is a constant real number $c$ for which:
\[|g_{f_{\text{receive}_k}^{(N)}}(\overrightarrow{x}) - g_{f_{\text{receive}_k}^{(N)}}(\overrightarrow{y})|\leq c\cdot | \overrightarrow{x} - \overrightarrow{y} |\]

holds, where $| \overrightarrow{x} - \overrightarrow{y} |$ should be interpreted as the distance between the endpoints of the two vectors. Returning to the observation that we made at the start of the proof, and the fact that $q(i)$ is a probability, we have:
\[\frac{C(\{send(m)\})}{C(I_{m})}\leq 1\]
and,
\[q(N\cdot p\cdot C(\{send(m)\}))\leq 1\]

Therefore:
\[\frac{C(\{send(m)\})}{C(I_{m})}\cdot q(N\cdot p\cdot C(\{send(m)\}))\leq 1\]

and multiplying the two sides by $r_{send}\cdot \hat{x}_i$ we have:
\[r_{send}\cdot \hat{x}_i\cdot \frac{C(\{send(m)\})}{C(I_{m})}\cdot q(N\cdot p\cdot C(\{send(m)\}))\leq r_{send}\cdot \hat{x}_i\]

So in order to establish the main result it suffices to show that:
\[|r_{send}\cdot\overrightarrow{x}_i - r_{send}\cdot\overrightarrow{y}_i|\leq c\cdot | \overrightarrow{x} - \overrightarrow{y} |\]

for any positive $r_{send}$ we have:
\[|r_{send}\cdot \overrightarrow{x}_i - r_{send}\cdot \overrightarrow{y}_i|= r_{send}\cdot |\overrightarrow{x}_i - \overrightarrow{y}_i|\]

and since by definition the distance must have the following property:
\[|\overrightarrow{x}_i - \overrightarrow{y}_i|\leq | \overrightarrow{x} - \overrightarrow{y} |\]

we have:
\[r_{send}\cdot|\overrightarrow{x}_i - \overrightarrow{y}_i|\leq r_{send}\cdot | \overrightarrow{x} - \overrightarrow{y} |\]

by transitivity of the inequality relation $\leq$ we have shown that at least for $c=r_{send}$:
\[|g_{f_{\text{receive}_k}^{(N)}}(\overrightarrow{x}) - g_{f_{\text{receive}_k}^{(N)}}(\overrightarrow{y})|\leq r_{send}\cdot | \overrightarrow{x} - \overrightarrow{y} |\]

The inequality holds and the function $g_{f_{\text{receive}_k}^{(N)}}$ is Lipschitz continuous.

\item[-] The function $g_{f_{\text{send}_i}^{(N)}}=r_{send}\cdot \hat{x}_i$, is Lipschitz continuous following the proof for the first case.
\end{itemize}
\end{enumerate}
\end{proof}

\subsection{Neighborhood discovery protocol}\label{subsec:multiple}
In this part we will study a slightly different version of the neighborhood discovery protocol in an ALOHA network. The discovery works as follows: every once in a while each node broadcasts a HELLO message to advertise its presence in the network. All the neighbors hearing this will respond with an acknowledgement. The sender follows a passive acknowledgement model, and upon receiving an acknowledgement which ensures its discovery by at least one neighbor, proceeds with its internal processing. Such protocols are essential building blocks of many algorithms such as routing in wireless ad hoc networks.

We consider every node to run an identical implementation of the neighborhood discovery protocol. An abstract transition system for every node's behavior is given in Fig.~\ref{fig:alohanode}.
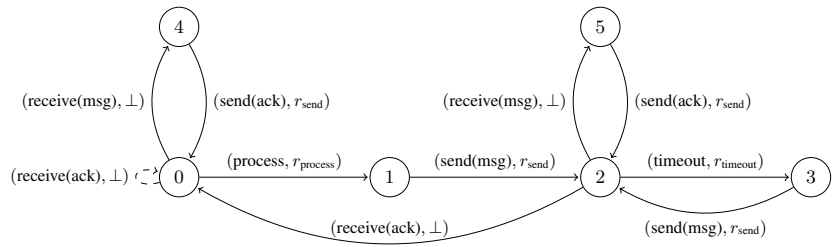
\begin{figure}[h]
\centering
\scalebox{0.8}{
\begin{tikzpicture}[shorten >=1pt,node distance=3.5cm,on grid, auto, every state/.style={inner sep=4pt,minimum size=3pt}]
   \node[state] (q_1) {$1$};
   \node[state] (q_0) [left=of q_1]  {$0$};
   \node[state] (q_2) [right=of q_1] {$2$};
   \node[state] (q_3) [above=2.5cm of q_2] {$5$};
   \node[state] (q_4) [above=2.5cm of q_0] {$4$};
   \node[state] (q_5) [right=of q_2] {$3$};
    \path[->]
    (q_0) edge  node {\scalebox{0.9}{$(\text{process},r_{\text{process}})$}} (q_1)
          edge [bend left] node {\scalebox{0.9}{$(\text{receive(msg)},\bot)$}} (q_4)
          edge [loop left, dashed] node {\scalebox{0.9}{$(\text{receive(ack)},\bot)$}} (q_0)
    (q_1) edge  node {\scalebox{0.9}{$(\text{send(msg)},r_{\text{send}})$}} (q_2)
    (q_2) edge [bend left] node [swap] {\scalebox{0.9}{$(\text{receive(ack)},\bot)$}} (q_0) 
          edge [bend left] node {\scalebox{0.9}{$(\text{receive(msg)},\bot)$}} (q_3)
          edge node {\scalebox{0.9}{$(\text{timeout}, r_{\text{timeout}})$}} (q_5)
    (q_3) edge [bend left] node {\scalebox{0.9}{$(\text{send(ack)},r_{\text{send}})$}} (q_2)
    (q_4) edge [bend left] node {\scalebox{0.9}{$(\text{send(ack)},r_{\text{send}})$}} (q_0)
    (q_5) edge [bend left] node {\scalebox{0.9}{$(\text{send(msg)},r_{\text{send}})$}} (q_2);
\end{tikzpicture}
}
\caption{Component transition system of ALOHA nodes in the neighborhood discovery protocol.} \label{fig:alohanode}
\end{figure}

In this transition system multiple assumptions and considerations are made. Receiving a message is allowed at any point in time when a node is not busy sending a message. Therefore when a node is doing its internal process, or when it is waiting for an acknowledgement, it responds to any message that is received in the mean-time.

Timeout, which is an internal action, is taken to have a fixed duration ($\mathcal{T}$) and the rate of the timeout is defined accordingly as ($r_{\textit{timeout}}=\frac{1}{\mathcal{T}}$). Nodes do other processing in between sending messages, which also happens after a fixed $\frac{1}{r_{\textit{process}}}$ time. A node which has a good performance spends most of its time processing in state $0$, instead of trying to send a message.

Based on the transition system in Fig.~\ref{fig:alohanode}, and also that $I_{msg}=\{msg, ack\}$ and $I_{ack}=\{msg, ack\}$, we derive the system of ODEs according to the recipe provided in Transformation~\ref{prop2}. The result is presented in Table~\ref{tbl:alohanodes}.

\begin{table}[h]
\caption{System of ODEs derived from the transition system in Fig. \ref{fig:alohanode}. $r_{ack}$ and $r_{msg}$ are only used for clarity and are otherwise equal to $r_{send}$.}\label{tbl:alohanodes}
\fbox{
\begin{minipage}{15em}
\begin{flalign*}
\frac{d\hat{x}_{0}}{dt}&= -r_{process}\cdot \hat{x}_{0} + r_{ack}\cdot \hat{x}_{4}&\\
& - r_{msg}\cdot(\frac{\hat{x}_{1}+\hat{x}_{3}}{\hat{x}_{1}+\hat{x}_{3}+\hat{x}_{4}+\hat{x}_{5}})\cdot q(N\cdot p\cdot(\hat{x}_{1}+\hat{x}_{3}+\hat{x}_{4}+\hat{x}_{5}))\cdot\hat{x}_{0}&\\
& + r_{ack}\cdot(\frac{\hat{x}_{4}+\hat{x}_{5}}{\hat{x}_{1}+\hat{x}_{3}+\hat{x}_{4}+\hat{x}_{5}})\cdot q(N\cdot p\cdot(\hat{x}_{1}+\hat{x}_{3}+\hat{x}_{4}+\hat{x}_{5}))\cdot\hat{x}_{2} &\\
\frac{d\hat{x}_{1}}{dt}&= -r_{msg}\cdot\hat{x}_{1}+ r_{process}\cdot \hat{x}_{0} &\\
\frac{d\hat{x}_{2}}{dt}&= - r_{timeout}\cdot \hat{x}_{2}+r_{msg}\cdot \hat{x}_{1}+r_{msg}\cdot \hat{x}_{3}+r_{ack}\cdot\hat{x}_{5}&\\
& - r_{msg}\cdot(\frac{\hat{x}_{1}+\hat{x}_{3}}{\hat{x}_{1}+\hat{x}_{3}+\hat{x}_{4}+\hat{x}_{5}})\cdot q(N\cdot p\cdot(\hat{x}_{1}+\hat{x}_{3}+\hat{x}_{4}+\hat{x}_{5}))\cdot \hat{x}_{2}&\\
& - r_{ack}\cdot(\frac{\hat{x}_{4}+\hat{x}_{5}}{\hat{x}_{1}+\hat{x}_{3}+\hat{x}_{4}+\hat{x}_{5}})\cdot q(N\cdot p\cdot(\hat{x}_{1}+\hat{x}_{3}+\hat{x}_{4}+\hat{x}_{5}))\cdot \hat{x}_{2}&\\
\frac{d\hat{x}_{3}}{dt}&=r_{timeout}\cdot \hat{x}_{2} - r_{msg}\cdot\hat{x}_{3}&\\
\frac{d\hat{x}_{4}}{dt}&=-r_{ack}\cdot \hat{x}_{4}&\\
& + r_{msg}\cdot(\frac{\hat{x}_{1}+\hat{x}_{3}}{\hat{x}_{1}+\hat{x}_{3}+\hat{x}_{4}+\hat{x}_{5}})\cdot q(N\cdot p\cdot(\hat{x}_{1}+\hat{x}_{3}+\hat{x}_{4}+\hat{x}_{5}))\cdot \hat{x}_{0}&\\
\frac{d\hat{x}_{5}}{dt}&=-r_{ack}\cdot \hat{x}_{5}&\\
& + r_{msg}\cdot(\frac{\hat{x}_{1}+\hat{x}_{3}}{\hat{x}_{1}+\hat{x}_{3}+\hat{x}_{4}+\hat{x}_{5}})\cdot q(N\cdot p\cdot(\hat{x}_{1}+\hat{x}_{3}+\hat{x}_{4}+\hat{x}_{5}))\cdot \hat{x}_{2}&\\
\end{flalign*}
\end{minipage}
}
\end{table}

We take the parameters ($r_{send}=100$, $r_{\mathit{process}}=1$, $r_{\mathit{timeout}}=30$) and $p=0.05$ for the connectivity of links in the network. We take $N=500$, and therefore we have the number of nodes in a neighborhood $d=N\cdot p=25$. By solving the system of equations for the initial conditions ($\hat{x}_{0}(0)=1$) and then for ($\hat{x}_{3}(0)=1$), we see a bistable behavior, as can be seen in Fig.~\ref{fig:ma} and \ref{fig:mb}. 
 
\begin{figure}[h]
\centering
\begin{subfigure}{0.45\textwidth}
\includegraphics[scale=0.43]{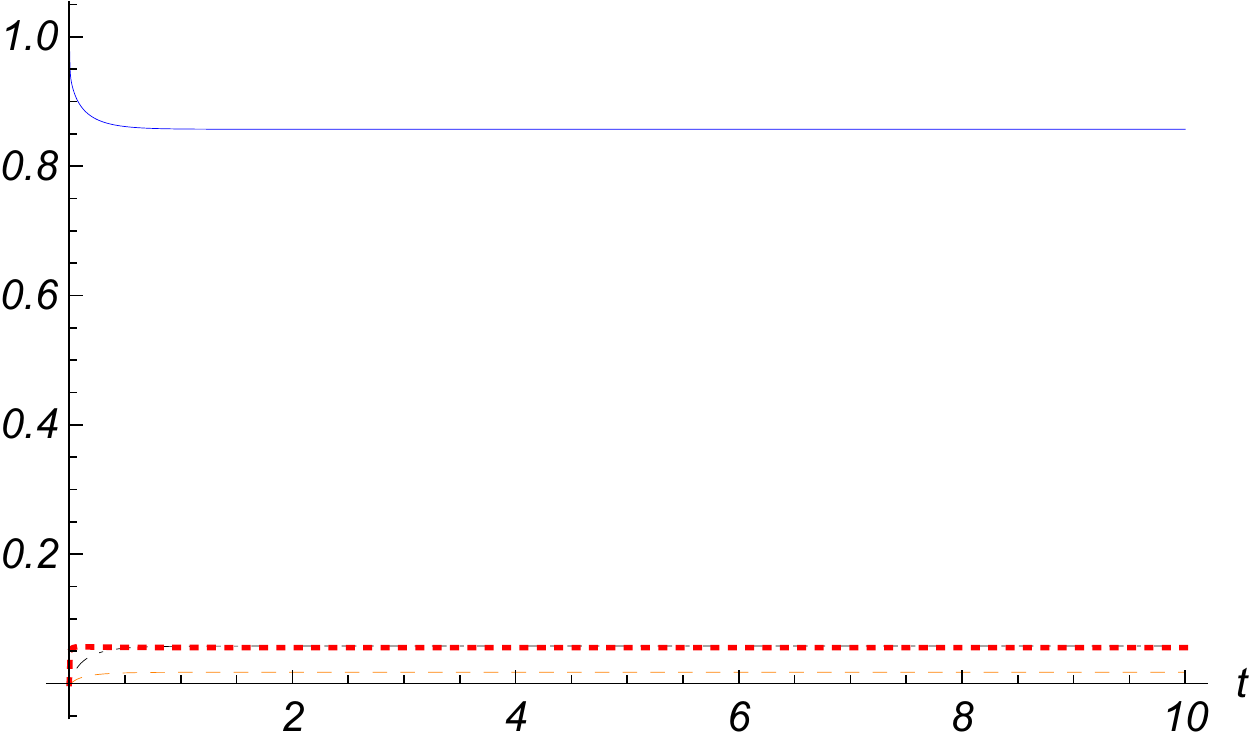}
\caption{Solution with $\hat{x}_0=1$} \label{fig:ma}
\end{subfigure}
\begin{subfigure}{0.45\textwidth}
\includegraphics[scale=0.6]{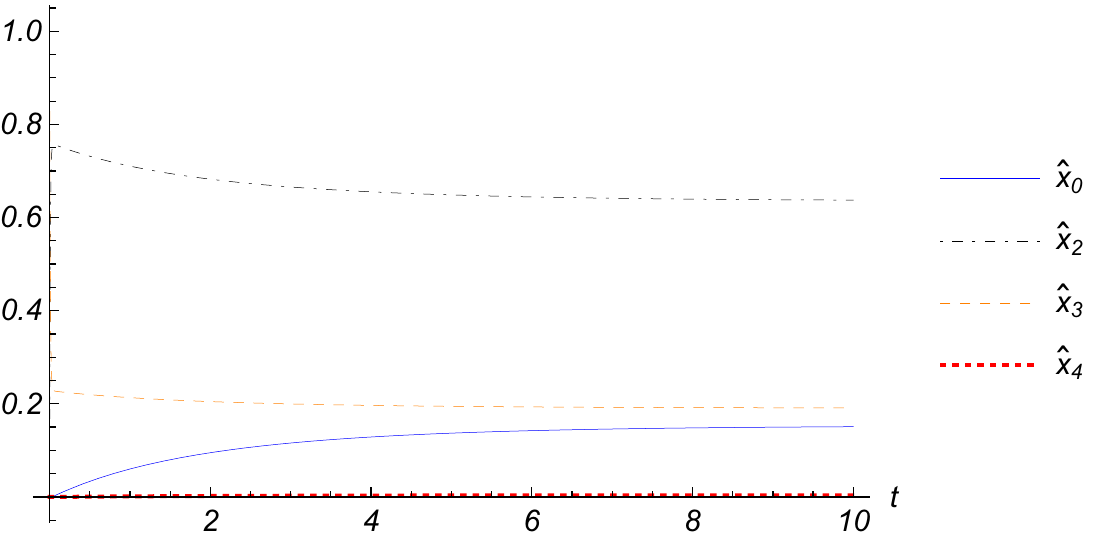}
\caption{Solution with $\hat{x}_3=1$} \label{fig:mb}
\end{subfigure}
\caption{The bistable behavior in an ALOHA network with multiple receivers} \label{fig:multireceiver}
\end{figure}
\section{Stability and Vector Field Analysis}\label{sec:vector}
A common method to study stability in Markov counting processes has been the calculation of \textit{drifts}. However, as the model grows, observing drifts in differential equations is not straightforward. Therefore we use vector fields to study the equilibrium points.

For equations in Table~\ref{tbl:simplealohanodes}, a corresponding two dimensional vector-field has been given in Fig~\ref{fig:va}. Here, since we know that $\hat{x}_{O}+\hat{x}_{T}+\hat{x}_{R}=1$, only the triangular area at the bottom of the plane is filled, where $\hat{x}_{O}+\hat{x}_{R}\leq 1$. The value of $\hat{x}_{T}$ at each point in this area is implicitly defined as $\hat{x}_{T}=1-\hat{x}_{O}-\hat{x}_{R}$. The vectors in this figure are of two colors, and show the tendency of the system starting from that point (as the initial condition) to go in any two directions: the red part consists of points which go to the first fixpoint, and the blue part consists of points which go to the second one.

For equations in Table~\ref{tbl:alohanodes}, Fig.~\ref{fig:vb} shows a similar pattern. Here, since there are 6 varying parameters $\hat{x}_{0}$ through $\hat{x}_{5}$, we only present two values for each point, one for parameter $\hat{x}_{0}$ which is the original state and the other for $\hat{x}_{3}$, the state in which a node retries sending a messages. Also we take $\hat{x}_{0}+\hat{x}_{3}\leq 1$, and therefore the rest of the parameters may take positive values which satisfies the equation $\hat{x}_{1}+\hat{x}_{2}+\hat{x}_{4}+\hat{x}_{5}=1-\hat{x}_{0}+\hat{x}_{3}$. For each point, Fig.~\ref{fig:vb} shows the average tendency to go to either of the two fixpoints, where white areas go to the first fixpoint and black areas go to the second one. For areas in which it is possible to go to both fixpoints for various values of the 4 absent parameters, the color will turn out to be different shades of gray, depending on the proportion of cases which go to each fixpoint.

In most systems, avoiding a bistable behavior improves the predictability of the system. However, when bistability is inevitable, a system's conditions can be monitored for signs that make it prone to slipping into a state with low performance, e.g. when the system in Fig.~\ref{fig:vb} is in the gray area.

\begin{figure}[h]
\centering
\begin{subfigure}{0.46\textwidth}
\includegraphics[scale=.44]{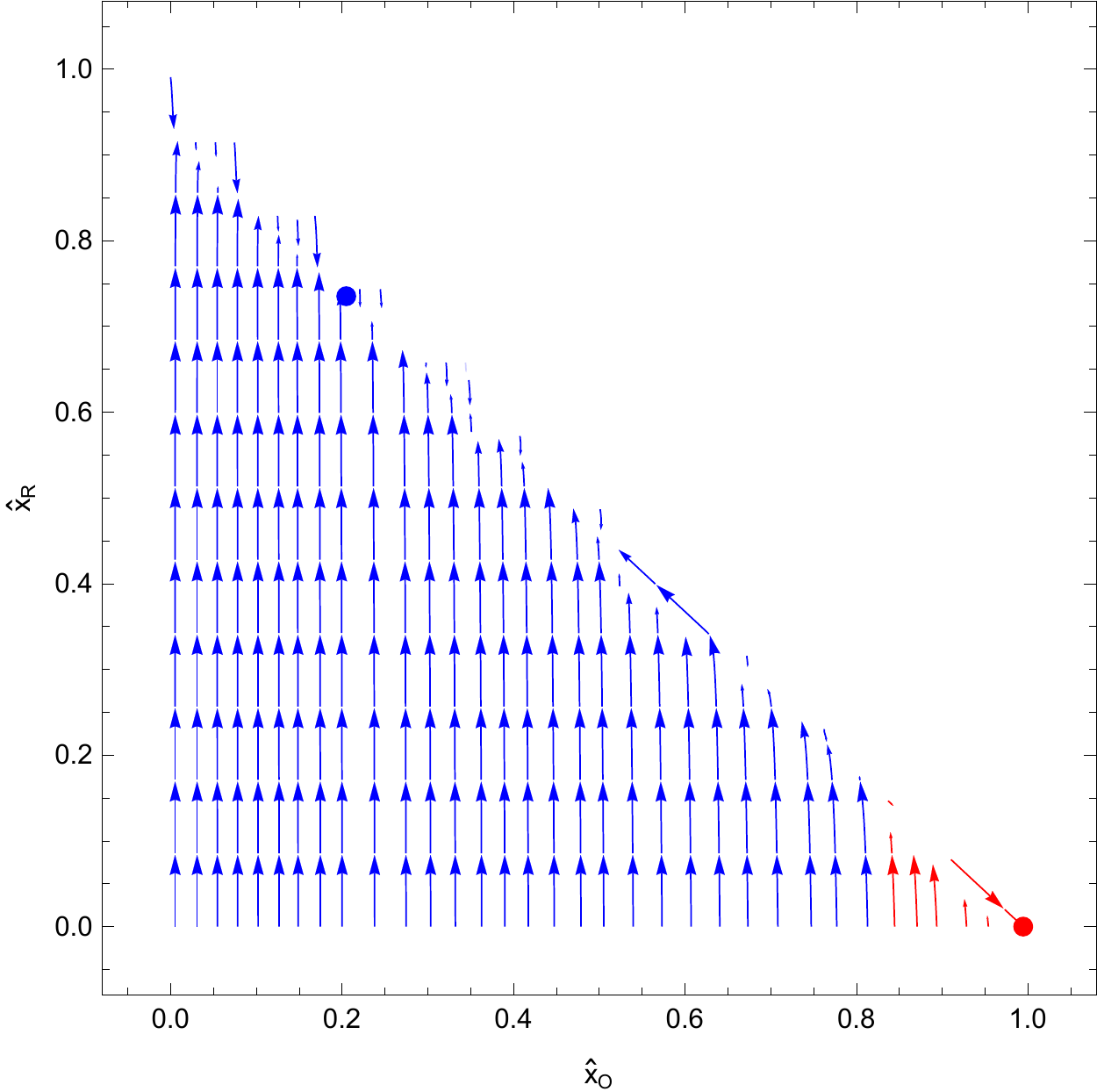}
\caption{Tendency of each $(\hat{x}_{O},\hat{x}_{R})$ pairs to go to either of the fixpoints (disks).} \label{fig:va}
\end{subfigure}
\begin{subfigure}{0.46\textwidth}
\includegraphics[scale=.44]{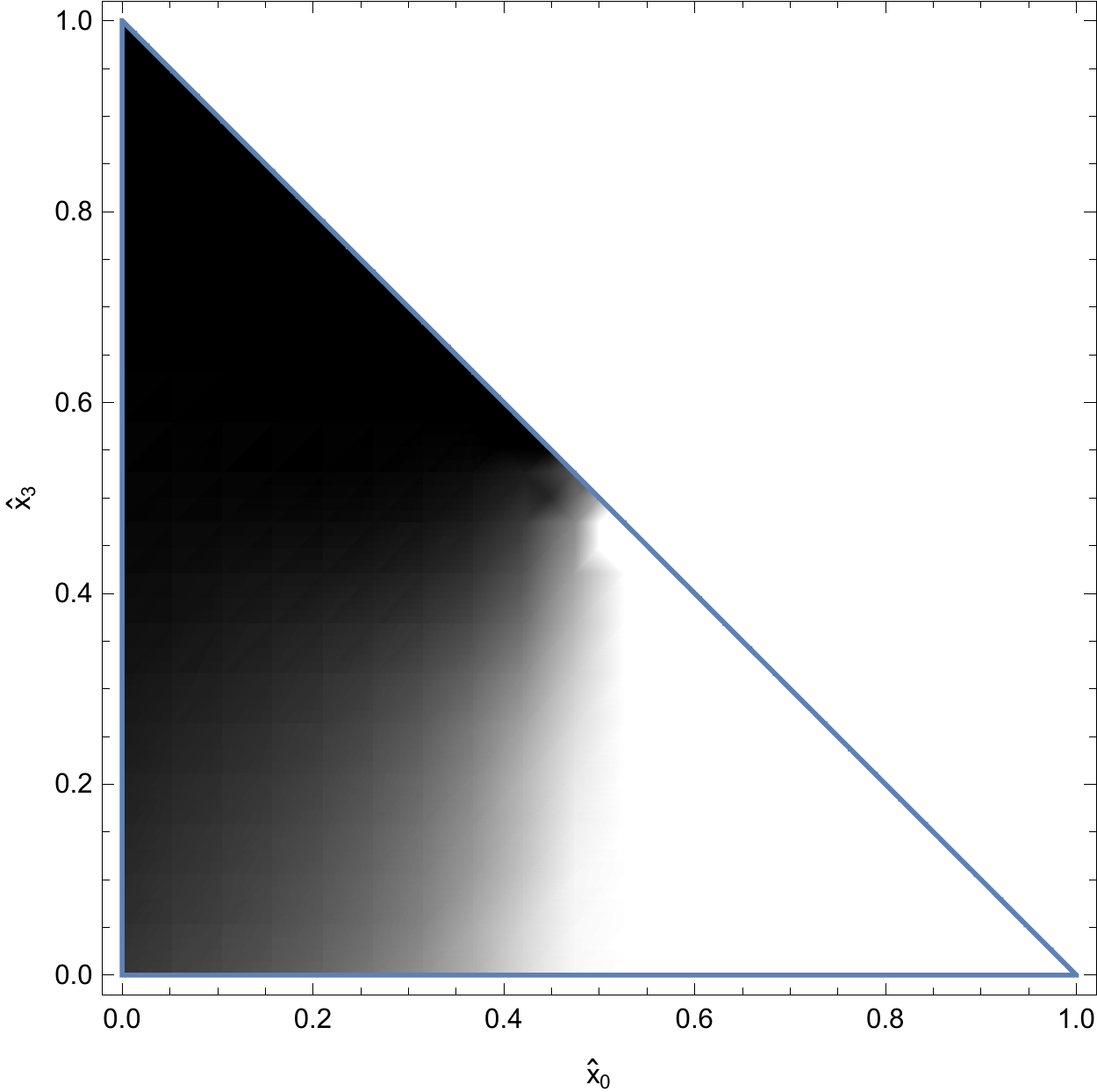}
\caption{Tendency of each $(\hat{x}_{0},\hat{x}_{3})$ pairs to go to either of the fixpoints.} \label{fig:vb}
\end{subfigure}
\caption{Visualization of bistable behavior in ALOHA} \label{fig:vectorfield}
\end{figure}

\section{Conclusion}
In this paper, we took crucial steps in modeling large wireless networks. We used insights from the field of wireless communication to describe communication models, and customize the mean-field theory. In this way we are able to reason about networks which are immensely large, and avoid common limitations in analyzing traditional models.

The reasons for the correctness of this modeling approach is twofold. First, we proposed semantics (Transformation~\ref{prop1} and \ref{prop2}) which fully conformed with theories that guarantee effective approximation of Markov chain models with systems of ODEs (Lemma \ref{lem:lem2} and \ref{lem:lem3}). Second, through examining ALOHA networks we were able to witness the same bistability phenomena that are commonly associated with them in practice. 

Finally, we demonstrated the modeling capabilities of this approach with a simple neighborhood discovery protocol. Moreover, we focused on systems with multiple equilibrium points and used vector-fields as a way to visualize the dynamic behavior of the system.

\paragraph{\textbf{Acknowledgments.}} The research from DEWI project (www.dewi-project.eu) leading to these results has received funding from the ARTEMIS Joint Undertaking under grant agreement No.~621353.

\bibliography{all}
\end{document}